\documentclass[11pt]{amsart}
\title{The Dirac oscillator, generalised parastatistics and colour Lie superalgebras}
\author[P.S.~Isaac]{Phillip S. Isaac}
\author[M.~Ryan]{Mitchell Ryan}

\usepackage[margin=1.13in]{geometry}
\usepackage[colorlinks=true, pdfstartview=FitV, linkcolor=blue, citecolor=blue, urlcolor=blue]{hyperref}
\usepackage[capitalize]{cleveref}
\usepackage{amsmath,amssymb,amsfonts,amsthm}
\usepackage{mathtools}
\usepackage[mathscr]{euscript}
\usepackage[foot]{amsaddr}
\usepackage{physics}
\usepackage{enumitem}
\usepackage{cite}
\usepackage[llbracket,rrbracket]{stmaryrd}
\SetSymbolFont{stmry}{bold}{U}{stmry}{m}{n}
\usepackage{tikz}
\usepackage{centernot}
\usetikzlibrary{arrows.meta,tikzmark}

\newcommand{\commabefore}[1]{%
\if\relax\detokenize{#1}\relax%
	#1%
\else%
	,#1%
\fi%
}
\newcommand{\commaafter}[1]{%
\if\relax\detokenize{#1}\relax%
	#1%
\else%
	#1,%
\fi%
}

\DeclareDocumentCommand\cbrak{ l m m }{\braces#1{\llbracket}{\rrbracket}{#2,#3}} 
\newcommand{\adjt}[1]{#1^{\dagger}}				
\newcommand{\acreate}{a^+}					
\newcommand{\aann}{a^-}						
\newcommand{\apm}[1][\pm]{a^{#1}}				
\newcommand{\bcreate}{b^+}					
\newcommand{\Bcreate}{\vec{b}^+}				
\newcommand{\bcreatespin}[1][]{\bcreate_{\sigma\commabefore{#1}}}
\newcommand{\bann}{b^-}						
\newcommand{\Bann}{\vec{b}^-}					
\newcommand{\bannspin}[1][]{\bann_{\sigma\commabefore{#1}}}
\newcommand{\bpm}[1][\pm]{b^{#1}}				
\newcommand{\bpmspin}[1][\pm]{\bpm[#1]_{\sigma}}		
\newcommand{\bfa}{\mathfrak{bf}}				
\newcommand{\ccreate}{c^+}					
\newcommand{\cann}{c^-}						
\newcommand{\cpm}[1][\pm]{c^{#1}}				
\newcommand{\CLS}{C_{LS}}					
\newcommand{\Hilb}{\mathcal{H}}					
\newcommand{\Hvac}{H_0}						
\newcommand{\HxH}{C^{\times}}					
\newcommand{\screate}{s^+}					
\newcommand{\Screate}{\vec{s}^+}				
\newcommand{\screatespin}{\screate_{\sigma}}			
\newcommand{\sann}{s^-}						
\newcommand{\Sann}{\vec{s}^-}					
\newcommand{\sannspin}{\sann_{\sigma}}				
\newcommand{\spm}[1][\pm]{s^{#1}}				
\newcommand{\spmspin}[1][\pm]{\spm[#1]_{\sigma}}		
\newcommand{\Sc}{{s_0}}						
\newcommand{\tcreate}{t^+}					
\newcommand{\tpm}[1][\pm]{t^{#1}}				
\newcommand{\CC}{\mathbb{C}}					
\newcommand{\Cl}[2][]{C\ell_{\commaafter{#1}#2}}		
\newcommand{\gen}[1]{\langle#1\rangle}				
\newcommand{\gl}{\mathfrak{gl}}					
\newcommand{\so}{\mathfrak{so}}					
\newcommand{\Hsch}{H_{\textup{Sch}}}				
\newcommand{\Id}{I}						
\newcommand{\labelian}{\mathfrak{a}}				
\renewcommand{\laplacian}{\Delta}				
\newcommand{\osp}{\mathfrak{osp}}				
\newcommand{\orth}{\mathfrak{o}}				
\newcommand{\pso}{\mathfrak{pso}}				
\newcommand{\spl}[1][]{\mathfrak{sl}}		 		
\newcommand{\vac}{\ket{0}}					
\newcommand{\vactilde}{\ket*{\widetilde{0}}}			
\renewcommand{\vec}[1]{\boldsymbol{\mathrm{#1}}}		
\newcommand{\ZZ}{\mathbb{Z}}					
\newcommand{\Ztwo}[1][]{\ZZ_2^{#1}}				
\newcommand{\Ztzt}{\Ztwo\times\Ztwo}				

\newcommand{\dfn}[1]{{\color{red!70!black}\itshape #1}}

\newtheorem{thm}{Theorem}[section]
\newtheorem{lem}[thm]{Lemma}
\newtheorem{prop}[thm]{Proposition}

\theoremstyle{definition}

\theoremstyle{remark}
\newtheorem{rmk}[thm]{Remark}

\usepackage[colorinlistoftodos]{todonotes}
\definecolor{lightblue}{RGB}{137,207,240}

\AddToHook{env/lem/begin}{\crefalias{thm}{lem}}
\AddToHook{env/prop/begin}{\crefalias{thm}{prop}}
\AddToHook{env/cor/begin}{\crefalias{thm}{cor}}

\AddToHook{env/defn/begin}{\crefalias{thm}{defn}}
\AddToHook{env/ex/begin}{\crefalias{thm}{ex}}
\AddToHook{env/rmk/begin}{\crefalias{thm}{rmk}}

\Crefname{thm}{Theorem}{Theorems}
\Crefname{lem}{Lemma}{Lemmas}
\Crefname{ex}{Example}{Examples}

\address[P.S.~Isaac]{School of Mathematics and Physics, University of Queensland, St.\ Lucia, QLD 4072, Australia}
\address[M.~Ryan]{Department of Mathematics, University of British Columbia, Vancouver, BC, Canada}
\email[P.S.~Isaac]{\href{mailto:psi@maths.uq.edu.au}{psi@maths.uq.edu.au}}
\email[M.~Ryan]{\href{mailto:mitchell.ryan@math.ubc.ca}{mitchell.ryan@math.ubc.ca}}

\keywords{Color Lie (super)algebras, Graded Lie (super)algebras, Dirac equation}
\subjclass[2020]{17B75, 17B70, 81Q05}

\begin{document}

\begin{abstract}
	We study the Dirac oscillator in one, two and three spatial dimensions, showing that the corresponding ladder operators realise the \( \mathbb{Z}_2\times\mathbb{Z}_2 \)-graded Lie superalgebras \( \mathfrak{pso}(3|2) \), \( \mathfrak{pso}(3|4) \) and \( \mathfrak{osp}_{01}(1|2) \oplus \mathfrak{sl}_{10}(1|1)\). These algebraic structures are related to parastatistics and their Fock spaces. We demonstrate that these colour algebras and Fock spaces are useful for analysing the Dirac oscillator and its eigenspaces, particularly in \( (1+3) \)-dimensions.
\end{abstract}

\maketitle

\section{Introduction}

In the last decade, there has been a resurgence of interest in colour Lie (super)algebras, particularly
associated with \( \Ztwo[n] \) colouring. 
This was initially due to a connection made in~\cite{Tolstoy2014b} between
\( \Ztzt \)-graded colour Lie superalgebras and the parastatistics of Green
\cite{Green1953,GreenbergMessiah1965}. 
Of particular note are the articles by Aizawa, Kuznetsova,
Tanaka and Toppan~\cite{AKTT2016,AKTT2017} that demonstrate the utility of \( \Ztzt \)-graded colour Lie superalgebras as dynamical symmetry algebras of the L\'evy-Lebond equation
\cite{LevyLeblond1967}, which is a non-relativistic analogue of the Dirac equation~\cite{Dirac1928}. 
Earlier work of Toppan~\cite{Toppan2015}
demonstrated, among other things, that Schr\"{o}dinger operators can be considered to have a
\( \Ztwo \)-graded superalgebra symmetry. This was further developed in~\cite{AKTT2016,AKTT2017}, where the L\'evy-Lebond
equation was obtained by effectively taking a square root of operators, which, loosely speaking, led to the \( \Ztzt \)-grade of the Schr\"{o}dinger case being doubled, i.e.\ giving rise to a \( \Ztzt \)-graded colour Lie superalgebra symmetry. Various applications of \( \Ztzt \)-graded colour Lie superalgebras have continued to be explored recently in the context of generalised quantum mechanics (see \cite{AIT2024,Bruce2024} and references therein).

A very recent study by one of the current authors~\cite{Ryan2025} revealed that there is a
\( \Ztzt \)-graded colour Lie superalgebra associated with the L\'evy-Lebond
equation generated by ladder operators. 
The resulting approach and methods presented in~\cite{Ryan2025} have further inspired the results of
the current work, which, as far as we know, have yet to be discussed in the literature of this burgeoning field. 
Our approach in the current article is to investigate the relativistic setting originally posed by Dirac, and to revisit graded
algebras of ladder operators through a modern coloured lense. 
Specifically, our focus is the Dirac oscillator in one, two and three spatial dimensions.
Besides the current work, colour Lie superalgebras have previously been utilised in a relativistic setting in~\cite{Vasiliev1985,Bruce2021,IN2025}.

The Dirac oscillator itself is a well-studied type of Dirac equation that involves linear terms
in both position and momentum, originally introduced by It\^o, Mori and
Carriere in~\cite{IMC1967}. It was later derived independently by Moshinsky and Szczepaniak
\cite{MoshinskySzczepaniak1989}, from which many works followed. In particular, exact solutions
have been derived and their properties explored in many key cases (see, for example,
\cite{MorenoZentella1989,BMyRNYSB1990,deLange1991,CFLU1991,Crawford1993,Villalba1994,QuesneTkachuk2005} and references
therein). Also, algebraic aspects of the Dirac oscillator related to supersymmetry and
supersymmetric quantum mechanics have been examined in a variety of contexts
\cite{QuesneMoshinsky1990,BeckersDebergh1990,Quesne1991,WWJ2021}. Applications of the Dirac
oscillator have been presented within the last 20 years in an assortment of fields, including quantum
optics~\cite{BMDS2007a,BMDS2007b}, an experimental microwave realisation~\cite{FVSBKMS2013},
dynamics of charge carriers in graphene~\cite{QuimbayStrange2013,Boumali2015} and more recently
on cosmic strings, topological defects and spacetime~\cite{BM2018,HHdeM2019,Ahmed2022} and coherent
states~\cite{NogamiToyama1996,GR2021}. 

In the current article, we present a confluence of the notion of parastatistics with the Dirac oscillator.
Following on from the work of Tolstoy~\cite{Tolstoy2014b}, there has been much recent interest in
parastatistics associated with \( \Ztzt \)-graded colour Lie superalgebras and
Fock (para)spaces~\cite{SVdJ2018,Toppan2021a,Toppan2021b,Zhang2023,SVdJ2024,SVdJ2025}.  In fact, we make use of the same algebraic structures used to study parabosons and parafermions in the works of Stoilova and Van der Jeugt \cite{SVdJ2005-1,SVdJ2005-2,SVdJ2008}; however, the corresponding statistics are distinct from parastatistics.  

An outline of the article is as follows. In \cref{sec:diracoscillator} we briefly review the Dirac oscillator and fix
our notation. \Cref{sec:parastatistics} then provides a brief summary of parastatistics, and give a relevant generalisation.
Sections~\ref{sec:1+1}, \ref{sec:1+2} and \ref{sec:1+3} then present the Dirac oscillator
in spatial dimensions 1, 2 and 3 respectively, including a presentation of the explicit ladder
operators, the related \( \Ztzt \)-graded colour Lie superalgebra in each
case, and present technical aspects of the determination of the spectrum in each case.

\section{Dirac Oscillator}\label{sec:diracoscillator}
Throughout this paper, we will work in natural units: \( \hbar = c = 1 \).
The Dirac oscillator is obtained by introducing a linear potential
via coupling \( \vec{p}\rightsquigarrow \vec{p} - \beta im\omega\vec{r} \) in the free Dirac equation;
explicitly, the Hamiltonian for the Dirac oscillator is
\begin{align*}
	H &= \vec{\alpha}\vdot(\vec{p} - \beta im\omega\vec{r}) + \beta m
\end{align*}
where \( \vec{r} = (x_1,\ldots,x_n) \) and \( \vec{p} = (-i\pdv{x_1},\ldots,-i\pdv{x_n}) \)
(position representation);
\( m \) is the mass; \( \omega \) is the oscillation frequency;
and the components of the vector operator \( \vec{\alpha} = (\alpha_1,\ldots,\alpha_n) \) and \( \beta \)
form generators of the Clifford algebra \( \Cl{n+1}(\CC) \), so that
\[
	\acomm{\alpha_i}{\alpha_j} = 2\delta_{ij}, \qquad
	\acomm{\alpha_i}{\beta} = 0, \qquad
	\acomm{\beta}{\beta} = 2.
\]
We use the convention that \( (x,y,z) = (x_1,x_2,x_3) \),
and use the notation \( \pdv{x} = \partial_x \), etc.
Squaring \( H \) and taking a non-relativistic limit yields the non-relativistic quantum harmonic oscillator with a strong spin-orbit coupling~\cite{MoshinskySzczepaniak1989}.

In one or two spacial dimensions,
the Clifford algebra relations can be realised with \( 2\times 2 \) matrices
(see e.g.~\cite{NogamiToyama1996,Villalba1994}),
and can be realised with \( 4\times 4 \) matrices in three spacial dimensions (the Dirac representation---see e.g.~\cite{MoshinskySzczepaniak1989}).
However, the results of this paper can be obtained independent of the choice of matrix representation.


\section{Parastatistics}\label{sec:parastatistics}
Recall that the creation and annihilation operators for fermions \( f_i^{\pm} \) (\( i=1,\ldots,m \)) and bosons \( b_i^{\pm} \) (\( i=1,\ldots,n \)) satisfy the following canonical (anti)commutation relations
\begin{align} 
	\begin{split}
		\acomm{f_i^+}{f_j^-} &= \delta_{ij}, \\
		\comm{b_i^+}{b_j^-} &= \delta_{ij},
	\end{split}
	\begin{split}
		\acomm{f_i^\pm}{f_j^\pm} &= 0, \\
		\comm{b_i^\pm}{b_j^\pm} &= 0.
	\end{split}
	\label{eq:fermionboson}
\end{align}
The corresponding free Hamiltonians are
\begin{equation}\label{eq:fermionbosonHamiltonian}
	H_f = \frac{1}{2}\sum_{j} \omega_j\comm{f_j^+}{f_j^-}
	\qquad \text{and} \qquad
	H_b = \frac{1}{2}\sum_{j} \omega_j\acomm{b_j^+}{b_j^-}
\end{equation}
for fermions and bosons respectively,
where the constants \( \omega_j \) represent the energy of a free particle in quantum state \( j \).
The fermion and boson operators are ladder operators for their respective free Hamiltonian:
\begin{align}
	\comm{H_f}{f_k^\pm} &= \pm \omega_k f_k^{\pm}, & \comm{H_b}{b_k^\pm} &= \pm \omega_k b_k^{\pm}.\label{eq:fermionbosonladder}
\end{align}

Parastatistics generalises the fermion and boson operators to
\dfn{parafermion} operators \( f_i^{\pm} \) and \dfn{paraboson} operators \( b_i^{\pm} \),
which need not satisfy the canonical (anti)commutation relations~\eqref{eq:fermionboson},
and instead satisfy~\cite{Green1953,GreenbergMessiah1965} the trilinear relations
\begin{align}
	\begin{split}
		\comm{\comm{f_i^+}{f_j^-}}{f_k^-} &= -2\delta_{ik}f_j^-, \\
		\comm{\acomm{b_i^+}{b_j^-}}{b_k^-} &= -2\delta_{ik}b_j^-,
	\end{split}
	\begin{split}
		\comm{\comm{f_i^-}{f_j^-}}{f_k^-} &= 0, \\
		\comm{\acomm{b_i^+}{b_j^-}}{b_k^-} &= 0.
	\end{split}
	\label{eq:parafermionbosonpartial}
\end{align}
Note that ordinary bosons and fermions satisfy the above relations,
but there exist operators which satisfy \eqref{eq:parafermionbosonpartial} but not \eqref{eq:fermionboson}.
These trilinear relations~\eqref{eq:parafermionbosonpartial} were chosen to ensure that the parafermion and paraboson operators remain ladder operators for their respective free Hamiltonians~\cite{Green1953}, i.e.\ still satisfy relations~\eqref{eq:fermionbosonladder}.

By applying the (super) Jacobi identity and taking adjoints of the parastatistics relations~\eqref{eq:parafermionbosonpartial} (noting that \( \adjt{(f_i^\pm)}=f_i^\mp \) and \( \adjt{(b_i^\pm)} = b_i^\mp \)),
we obtain relations for all possible triples of parafermions or parabosons.
We can express all such relations with the following two expressions
\begin{align}
	\comm{\comm{f_i^{\zeta}}{f_j^{\eta}}}{f_k^{\xi}} &= \abs{\xi - \eta}\delta_{jk}f_i^\zeta - \abs{\xi - \zeta}\delta_{ik}f_j^{\eta}\label{eq:parafermion},\\
	\comm{\acomm{b_i^{\zeta}}{b_j^{\eta}}}{b_k^{\xi}} &= (\xi - \eta)\delta_{jk}b_i^{\zeta} + (\xi - \zeta)\delta_{ik}b_j^{\eta}\label{eq:paraboson}
\end{align}
where \( \zeta,\eta,\xi\in\{+,-\} \) are interpreted as \( +1 \) and \( -1 \) in the above equations.
The parafermion relations~\eqref{eq:parafermion} give rise to the Lie algebra \( \orth(2m+1) \)~\cite{KamefuchiTakahashi1962} while the paraboson relations~\eqref{eq:paraboson} relations give rise to the Lie superalgebra~\( \osp(1|2n) \)~\cite{GanchevPalev1980}.

When considering a mixed system of parabosons and parafermions,
Greenberg and Messiah~\cite{GreenbergMessiah1965} found that there were four types of (anti)commutation rules between parafermions and parabosons: the two trivial cases where the paraboson operators commute or anticommute with the parafermion operators;
the relative parafermion relations
\begin{align}\label{eq:relativeparafermion}
	\begin{split}
		\comm{\comm{f_i^{\zeta}}{f_j^{\eta}}}{b_k^{\xi}} &= 0,\\
		\comm{\comm{f_i^{\zeta}}{b_j^{\eta}}}{f_k^{\xi}} &= -\abs{\xi-\zeta}\delta_{ik}b_j^{\eta},
	\end{split}
	\begin{split}
		\comm{\acomm{b_i^{\zeta}}{b_j^{\eta}}}{f_k^{\xi}} &= 0,\\
		\acomm{\comm{f_i^{\zeta}}{b_j^{\eta}}}{b_k^{\xi}} &= (\xi-\eta)\delta_{jk}f_i^{\zeta}
	\end{split}
\end{align}
and the relative paraboson relations
\begin{align}\label{eq:relativeparaboson}
	\begin{split}
		\comm{\comm{f_i^{\zeta}}{f_j^{\eta}}}{b_k^{\xi}} &= 0, \\
		\acomm{\acomm{f_i^{\zeta}}{b_j^{\eta}}}{f_k^{\xi}} &= \abs{\xi-\zeta}\delta_{ik}b_j^{\eta},
	\end{split}
	\begin{split}
		\comm{\acomm{b_i^{\zeta}}{b_j^{\eta}}}{f_k^{\xi}} &= 0,\\
		\comm{\acomm{f_i^{\zeta}}{b_j^{\eta}}}{b_k^{\xi}} &= (\xi-\eta)\delta_{jk}f_i^{\zeta}.
	\end{split}
\end{align}
Of particular note is that the relative paraboson relations~\eqref{eq:relativeparaboson},
together with~\eqref{eq:parafermion} and~\eqref{eq:paraboson},
give rise to the colour Lie superalgebra%
\footnote{We are using the \( \pso \) notation from~\cite{SVdJ2018}. In the notation of~\cite{Tolstoy2014b}, \( \pso(2m+1|2n) \) is written as \( \osp(1,2m|2n,0) \).}%
~\( \pso(2m+1|2n) \)~\cite{Tolstoy2014b}.
The link between colour Lie superalgebras and parastatistics has since been investigated further via the construction of  \( \Ztzt \)-graded parastatistic Fock spaces and para-spaces in the context of \( \Ztzt \)-graded quantum mechanics ~\cite{SVdJ2018,Toppan2021a,Toppan2021b,Zhang2023,SVdJ2024,SVdJ2025}.

The parastatistics relations~\eqref{eq:parafermionbosonpartial}
were chosen to ensure that the parastatistics operators are ladder operators for the corresponding free Hamiltonian (i.e.\ satisfy~\eqref{eq:fermionbosonladder}).
However, relations~\eqref{eq:parafermionbosonpartial}
are not the only choice of trilinear relations that yield ladder operators.
Indeed, the following parafermion-like and paraboson-like trilinear relations
\begin{align}
	\comm{\comm{f_i^+}{f_j^-}}{f_k^-} &= -2\delta_{ij}f_k^- + 2\delta_{jk}f_i^- - 2\delta_{ki}f_j^-, &
	\comm{\comm{f_i^-}{f_j^-}}{f_k^-} &= 0, \label{eq:notparafermion}\\
\comm{\acomm{b_i^+}{b_j^-}}{b_k^-} &= -2\delta_{ij}b_k^- + 2\delta_{jk}b_i^- - 2\delta_{ki}b_j^-, &
	\comm{\acomm{b_i^-}{b_j^-}}{b_k^-} &= 0 \label{eq:notparaboson}
\end{align}
ensure that \( \comm{H_f}{f_k^\pm} = \pm\omega f_k^\pm  \) and \( \comm{H_b}{b_k^\pm} = \pm\omega b_k^\pm \),
where \( H_f \) and \( H_b \) are defined as in~\eqref{eq:fermionbosonHamiltonian} with \( \omega_j = \omega/m \) and \( \omega_j = \omega/n \) (respectively) for every \( j \).
Relations between all other triples of operators appearing in either~\eqref{eq:notparafermion} or~\eqref{eq:notparaboson} can be obtained via the (super) Jacobi identity and by taking adjoints.
As we will see in~\cref{sec:1+2notparastatics} and~\cref{rmk:1+3notparastatistics},
the ladder operators for the \( (1+2) \)- and \( (1+3) \)-dimensional Dirac oscillator satisfy  relations~\eqref{eq:notparafermion} and~\eqref{eq:notparaboson}.

\begin{rmk}
	The operators satisfying relations~\eqref{eq:notparafermion} or~\eqref{eq:notparaboson} can only change the energy eigenvalues of \( H_f \) or \( H_d \) by the same amount \( \omega \) for each \( j \). This is more restrictive than the parastatistics relations~\eqref{eq:parafermionbosonpartial}, for which each parafermion or paraboson operator can change the energy eigenvalues by different amounts \( \omega_j \) depending on \( j \).
\end{rmk}

\section{Dirac Oscillator in (1+1)-dimensions}\label{sec:1+1}
The Dirac Oscillator Hamiltonian in \( (1+1) \)-dimensions is as follows
\begin{equation}\label{eq:1+1Hamiltonian}
	H = \alpha(-i\partial_x - \beta im\omega x) + \beta m
\end{equation}
where \( \alpha \) and \( \beta \) are gamma matrices (elements of the Clifford Algebra \( \Cl{2}(\CC) \)) satisfying
\[
	\acomm{\alpha}{\alpha}=\acomm{\beta}{\beta} = 2\Id, \qquad \acomm{\alpha}{\beta} = 0.
\]
The square of this Hamiltonian is
\[
	H^2 = -(\partial_x)^2 + m^2\omega^2x^2 + m^2 - \beta m\omega.
\]
Note that the familiar Harmonic oscillator Hamiltonian for the Schr\"odinger equation 
\begin{equation}\label{eq:Hsch1+1}
	\Hsch \coloneq \frac{1}{2m}(-i\partial_x)^2 + \frac{1}{2}m\omega^2x^2
\end{equation}
appears as a term in \( H^2 \).

By direct computation, we can verify that \( H^2 \) has the following ladder operators,
\begin{equation}\label{eq:1+1H2ladder}
	\begin{aligned}
		\bann &= -\sqrt{\frac{m\omega}{2}}x - \frac{1}{\sqrt{2m\omega}}\partial_x, & \comm{H^2}{\bann} &= -2m\omega\,\bann, \\
		\bcreate &= -\sqrt{\frac{m\omega}{2}}x + \frac{1}{\sqrt{2m\omega}}\partial_x, & \comm{H^2}{\bcreate} &= 2m\omega\,\bcreate,\\
		\sann &= -\frac{i}{2}(\beta \alpha + \alpha), & \comm{H^2}{\sann} &= -2m\omega\, \sann,\\
		\screate &= -\frac{i}{2}(\beta \alpha - \alpha), & \comm{H^2}{\screate} &= 2m\omega\,\screate.
	\end{aligned}
\end{equation}
These ladder operators for \( H^2 \) have the following relations with the Hamiltonian \( H \) itself:
\begin{align}
	\acomm{H}{\spm} &= \sqrt{2m\omega}\,\bpm, &
	\comm{H}{\bpm} &= \pm\sqrt{2m\omega}\,\spm.
	\label{eq:1+1interleave}
\end{align}

Notice that \( \bann \) and \( \bcreate \) are the familiar ladder operators for the non-relativistic Schr\"odinger harmonic oscillator.
As such, \( \bann,\bcreate \) satisfy the boson commutation relation
\[
	\comm{\bann}{\bcreate} = \Id.
\]
After choosing a representation, we can interpret \( \sann \) as shifting the small component of a spinor to the large component (hence, decreasing the energy of the spinsor) and interpret \( \screate \) as shifting the large component to the small component (hence, increasing the energy level of the spinor).

\begin{rmk}
	In the Dirac representation,
	\( \sann = \Sigma_- S_1 \) and \( \screate = \Sigma_+ S_1 \)
	where \( S_1 = \Id_{2\times 2}\otimes \sigma_1 \) is spin in the \( x \)-direction and 
	\( \Sigma_\pm = \sigma_\pm\otimes\Id_{2\times 2} \) are the so-called \( * \)-spin operators defined in~\cite{STS2010}.
\end{rmk}

The operators \( \sann,\screate \) satisfy the fermion anticommutation relations
\[
	\acomm{\sann}{\screate} = \Id, \qquad \acomm{\sann}{\sann} = \acomm{\screate}{\screate} = 0
\]
and the two pairs of ladder operators (\( \bann,\bcreate \) and \( \sann, \screate \)) mutually commute:
\[
	\comm{\bann}{\sann} = \comm{\bann}{\screate} = \comm{\bcreate}{\sann} = \comm{\bcreate}{\screate} = 0.
\]
Thus, we can interpret the operators \( \bann,\bcreate,\sann,\screate \) as a mixed boson and fermion system.
This interpretation gives rise to the boson-fermion Lie superalgebra  with the following basis:
\[
	0\text{-sector: } \Id,\bann,\bcreate, \qquad 1\text{-sector: } \sann,\screate.
\]
We will denote this Lie superalgebra by \( \bfa(1|1) \).
We could also extend the above Lie superalgebra to \( \bfa(1|1)\rtimes \gen{H^2,H} \) where \( H^2 \) is \( 0 \)-graded and \( H \) is \( 1 \)-graded.

\subsection{The colour Lie superalgebra}
Bosons and fermions automatically satisfy the paraboson~\eqref{eq:paraboson} and parafermion~\eqref{eq:parafermion} equations (respectively).
In addition, we can verify that \( \bann,\bcreate \) and \( \sann,\screate \) also satisfy the relative paraboson relations~\eqref{eq:relativeparaboson}.
The ladder operators \( \bann,\bcreate,\sann,\screate \) satisfying \cref{eq:paraboson,eq:parafermion,eq:relativeparaboson} generate the \( \Ztzt \)-graded colour Lie superalgebra%
\footnote{We are using the \( \pso \) notation from~\cite{SVdJ2018}. In the notation of~\cite{Tolstoy2014b}, \( \pso(3|2) \) is written as \( \osp(1,2|2,0) \).}
\( \pso(3|2) \)~\cite{Tolstoy2014b}.
This \( \Ztzt \)-graded colour Lie superalgebra has the following basis:
\begin{equation}\label{eq:pso32realisation}
	\begin{aligned}
		00\text{-sector: }& \Hsch, \Sc, (\bann)^2,(\bcreate)^2, & 01\text{-sector: }& \bann\sann, \bcreate\screate, \bann\screate, \bcreate\sann, \\
		10\text{-sector: }& \bann, \bcreate, & 11\text{-sector: }& \sann, \screate
	\end{aligned}
\end{equation}
and satisfies the following relations (with \( \Hsch \) given by~\eqref{eq:Hsch1+1} and \( \Sc \coloneqq -(1/2)\beta \)):
\begin{align*}
	\acomm{\bann}{\bcreate} &= 2\Hsch, &
	\comm{\sann}{\screate} &= -2\Sc,\\
	\acomm{\bann}{\sann} &= 2\bann\sann\text{, etc.}, &
	\comm{\Hsch}{\Sc} &= 0,\\
	\comm{\Hsch}{\bpm} &= \pm\omega\bpm, &
	\comm{\Hsch}{\spm} &= 0, \\
	\comm{\Sc}{\bpm} &= 0, &
	\comm{\Sc}{\spm} &= \pm\spm, \\
	\comm{\bpm[\eta]\spm[\zeta]}{\bpm[\xi]} &= \frac{1}{2}(\xi - \eta) \spm[\zeta], &
	\acomm{\bpm[\eta]\spm[\zeta]}{\spm[\xi]} &= \frac{1}{2}\abs{\xi - \zeta}\bpm[\eta]
\end{align*}
where \( \zeta,\eta,\xi\in\{+,-\} \) are interpreted as \( +1 \) and \( -1 \) in the above equations.
The remaining relations can be obtained via the colour Jacobi identity.
In particular, notice that this colour Lie superalgebra realises the relations
for \( \acomm{\bann}{\bcreate} \) and \( \comm{\sann}{\screate} \),
as opposed to the relations for \( \comm{\bann}{\bcreate} \) and \( \acomm{\sann}{\screate} \) of the Lie superalgebra \( \bfa(1|1) \).
Despite this, \( \pso(3|2) \) and \( \bfa(1|1) \) share the same Fock space in this representation.

One advantage of the \( \pso(3|2) \) colour Lie superalgebra is that it is semisimple (moreover it is simple and basic),
whereas the Lie superalgebra \( \bfa(1|1) \) is not.
This fact was exploited to study the Fock spaces of the \( \pso \) series in~\cite{SVdJ2018}.
In addition,
the \( \pso(3|2) \) algebra realises the \( \so(3) \) Lie algebra \( \gen{\Sc, \sann, \screate} \) as a subalgebra, whereas \( \bfa(1|1) \) does not.
(Note however, that \( \pso(3|2) \) realises a \emph{graded} version of \( \so(3) \).)

The square of the Hamiltonian \( H^2 = 2m\Hsch - 2m\omega\Sc + m^2\Id \) is 
\emph{not} a member of \( \pso(3|2) \) but \emph{is} a member of \( \pso(3|2)\oplus\CC\Id \) and is \( 00 \)-graded in this latter algebra.

\subsection{Including the Hamiltonian as a homogeneous element}\label{sec:makeHhomogeneous}
Notice that
\[
	H = \sqrt{2m\omega}(\bcreate\sann + \bann\screate) - 2m\Sc
\]
is an element of \( \pso(3|2) \),
but is not homogeneous.
Consequently, the ladder operator relations with \( H \) in~\eqref{eq:1+1interleave} are not realised.


To remedy this,
we can instead consider the shifted Hamiltonian
\[
	\Hvac \coloneq H - m\beta = \sqrt{2m\omega}(\bcreate\sann + \bann\screate)
\]
which is \( 01 \)-graded.
This Hamiltonian \( \Hvac \) satisfies the same ladder operator relations~\eqref{eq:1+1interleave} as \( H \),
and hence realises these relations within the colour Lie superalgebra \( \pso(3|2) \).
In fact, this shifted Hamiltonian is useful in its own right, for example,
the \( (1+2) \)-dimensional shifted Hamiltonian is used to model monolayer graphene~\cite{QuimbayStrange2013,Boumali2015}.

The main difference between \( \Hvac \) and \( H \) is that the vacuum state for \( \Hvac \) has energy \( 0 \) whereas the vacuum state for \( H \) has energy \( m \)
(see \cref{lem:betavac,lem:H0vac} below).

This shifted Hamiltonian \( \Hvac \) has the following relations within the colour Lie superalgebra
\begin{align*}
	\acomm{\Hvac}{\Hvac} &= 4m\Hsch + 4m\omega \Sc = 2\Hvac^2,\\
	\comm{\Hvac}{\bpm} &= \pm\sqrt{2m\omega}\,\spm,\\
	\acomm{\Hvac}{\spm} &= \sqrt{2m\omega}\,\bpm,\\
	\comm{\Hsch}{\Hvac} &= \sqrt{2m\omega}(\bcreate\sann - \bann\screate),\\
	\comm{\Sc}{\Hvac} &= -\sqrt{2m\omega}(\bcreate\sann - \bann\screate).
\end{align*}

\begin{rmk}
	Considering the shifted Hamiltonian is not the only way to resolve the non-homogeneity of \( H \).
	One alternative is to distinguish the two appearances of \( \beta \) in the Hamiltonian by enlarging the Clifford algebra from \( \Cl{2}(\CC) \) to \( \Cl{3}(\CC) \). Another alternative is to use a non-faithful non-graded representation of \( \pso(3|2) \) onto the Hilbert space.

	In general, \( H \) need not be homogeneous (even though a homogeneous Hamiltonian would be useful).
	We will see in \cref{sec:1+3} that, in \( (1+3) \)-dimensions, neither that Hamiltonian nor its square can be homogeneous.
\end{rmk}

\subsection{Fock space and the spectrum}

Here, we give details about the spectrum of \( H \), not only for completeness, but also for insight into the higher dimensional cases presented in later sections. For this \( (1+1) \)-dimensional case, our approach here derives the spectrum without making much use of the \( \Ztzt\)-graded colour superalgebra, unlike the \( (1+3) \)-dimensional case studied later in the article. 

Since \( H^2 \) is the square of the self adjoint operator \( H \), its spectrum must be non-negative.
Thus, there must exist a vacuum state \( \vac \) such that \( \bann\vac = \sann\vac = 0 \) 
(otherwise we could continually apply \( \bann \) or \( \sann \) to obtain negative eigenvalues).

\begin{lem}\label{lem:betavac}
	\( \beta\vac = \vac \)
\end{lem}
\begin{proof}
	\begin{align*}
		\beta \vac
		&= \comm{\sann}{\screate}\vac\\
		&= \sann\screate\vac - \screate\sann\vac\\
		&= -\screate\sann\vac + \acomm{\sann}{\screate}\vac - \screate\sann\vac\\
		&= \vac \qedhere
	\end{align*}
\end{proof}

\begin{lem}\label{lem:H0vac}
	\( \Hvac\vac = 0 \)
\end{lem}
\begin{proof}
	Observe that \( \Hvac = \sqrt{m\omega/2}(\acomm{\bcreate}{\sann} + \acomm{\bann}{\screate}) \).
	Since \( \comm{\bcreate}{\sann} = \comm{\bann}{\screate} = 0 \), we have that
	\( \Hvac = \sqrt{2m\omega}(\bcreate\sann + \screate\bann) \).
	Therefore, \( \Hvac\vac = 0 \).
\end{proof}
Notice that the above proofs used both relations from \( \bfa(1|1) \) and \( \pso(3|2) \).

Combining the above two lemmas, we find that
\[
	H\vac = \Hvac\vac + m\beta\vac = m\vac.
\]


Now, consider an eigenstate \( \ket{\psi} \) with eigenvalue \( E \).
Then
\begin{align}\label{eq:1+1raise}
	\ket{\psi^{\pm}} 
	&= (E\pm \sqrt{E^2 + 2m\omega})\bcreate\ket{\psi} + \sqrt{2m\omega}\,\screate\ket{\psi}
\end{align}
is a (normalisable) eigenvector of \( H \) with eigenvalue \( \pm\sqrt{E^2 + 2m\omega} \).
This is easily verified using the relations of \( H \) with \( \bcreate \) and \( \screate \).
We conclude that the spectrum of \( H \) is given by
\[
	\{m\}\cup\{\pm\sqrt{m^2 + 2nm\omega} \mid n~\text{a positive integer}\}
\]
(compare with the spectrum in~\cite{NogamiToyama1996} for example).

%

\section{Dirac Oscillator in (1+2)-dimensions}\label{sec:1+2}
The Dirac Oscillator Hamiltonian in \( (1+2) \)-dimensions is as follows
\begin{equation}\label{eq:1+2Hamiltonian}
	H = \alpha_1(-i\partial_x - \beta im\omega x) + \alpha_2(-i\partial_y - \beta im\omega y) + \beta m
\end{equation}
where \( \alpha_1,\,\alpha_2 \) and \( \beta \) are gamma matrices (elements of the Clifford Algebra \( \Cl{3}(\CC) \)) satisfying
\[
	\acomm{\gamma_j}{\gamma_k} = 2\delta_{jk}
	\qquad
	\text{for}~(\gamma_1,\gamma_2,\gamma_3) = (\alpha_1,\alpha_2,\beta).
\]

The square of this Hamiltonian is
\begin{align*}
	H^2 &= -\laplacian + m^2\omega^2r^2 + m^2 - 2m\omega\beta + (2m\omega x\partial_y - 2m\omega y\partial_x)\beta\alpha_1\alpha_2
\end{align*}
where \( \laplacian = (\partial_x)^2 + (\partial_y)^2 \) and \( r^2 = x^2 + y^2 \).

In \( (1+2) \)-dimensions, angular momentum appears (both orbital and intrinsic).
We shall treat \( S_0 \coloneq -(1/2)i\alpha_1\alpha_2 \) as the total spin
(compare with the spin operators~\eqref{eq:1+3spin} in \( (1+3) \)-dimensions).

The orbital angular momentum \( L_0 = x(-i\partial_y) - y(-i\partial_x) \) and spin \( S_0 = (-1/2)i\alpha_1\alpha_2 \) both commute with \( H^2 \).
Neither \( L_0 \) nor \( S_0 \) commute with \( H \),
but the total angular momentum \( J = L_0 + S_0 \) does: \( \comm{H}{J} = 0 \).
In addition, the operator \( \beta S_0 \) commutes with \( H \).

One of the main distinctions between the Dirac oscillator in \( (1+2) \)-dimensions, compared to \( (1+1) \)-dimensions, is the appearance of a spin-orbit coupling term in \( H^2 \)
\[
	\CLS = -2\beta L_0 S_0 - (1/2)\beta = - \frac{1}{2}\beta + (x\partial_y - y\partial_x)\beta\alpha_1\alpha_2.
\]
This specific spin-orbit coupling \( \CLS \) was chosen for the fact that \( \comm{H}{\CLS}=0 \).
Hence \( H \) and \( \CLS \) are simultaneously diagonalisable.

\subsection{Ladder operators}
Set \( (x_1,x_2) = (x,y) \).
Similar to \( (1+1) \)-dimensions, define the following operators
\begin{align*}
	\sann_j &= -\frac{i}{2}(\beta\alpha_j + \alpha_j), &
	\screate_j &= -\frac{i}{2}(\beta\alpha_j - \alpha_j),\\
	\bann_j &= -\sqrt{\frac{m\omega}{2}}x_j - \sqrt{\frac{1}{2m\omega}}\partial_{x_j}, &
	\bcreate_j &= -\sqrt{\frac{m\omega}{2}}x_j + \sqrt{\frac{1}{2m\omega}}\partial_{x_j}
\end{align*}
for \( j=1,2 \).
Observe that the corresponding number operator
\[
	N = 
	\sum_{k=1}^2 \left(\bcreate_k\bann_k + \frac{1}{2}\screate_k\sann_k\right)\\
\]
commutes with \( H \).
That is, \( H \) and \( N \) are simultaneously diagonalisable,
so we should consider the Fock space of \( N \) when solving this Dirac equation.
In fact, \( \sann_j, \screate_j \) are ladder operators for \( H^2 \), with
\[
	\comm{H^2}{\sann_j} = -4m\omega\, \sann_j, \qquad
	\comm{H^2}{\screate_j} = 4m\omega\, \screate_j
\]
However, unlike in \( (1+1) \)-dimensions,
the operators \( \bann_j,\bcreate_j \) do \emph{not} form ladder operators for the \( (1+2) \)-dimensional equation.
Instead, consider the following operators:
\begin{align*}
	\aann_j &= \frac{1}{\sqrt{2}}(\bann_j + \beta \alpha_j \alpha_{j+1} \bann_{j+1}) &
	\acreate_j &= \frac{1}{\sqrt{2}}(\bcreate_j - \beta \alpha_j\alpha_{j+1} \bcreate_{j+1})\\
	\cann_j &= \frac{1}{\sqrt{2}}(\bann_j - \beta \alpha_j \alpha_{j+1} \bann_{j+1}) &
	\ccreate_j &= \frac{1}{\sqrt{2}}(\bcreate_j + \beta \alpha_j\alpha_{j+1} \bcreate_{j+1})
\end{align*}
for \( j=1,2 \) (where, by convention, \( 2+1 = 1 \)).

\begin{rmk}
	The operators \( \apm_j \) and \( \cpm_j \) (\( j=1,2 \))
	are similar to the chiral operators from~\cite{BMDS2007b}.
\end{rmk}

The four operators \( \apm_j \) (\( j=1,2 \))
are ladder operators for both \( H^2 \), \( N \) and \( \CLS \):
\begin{align*}
	\comm{H^2}{\apm_j} &= \pm4m\omega\,\apm_j &
	\comm{N}{\apm_j} &= \pm\apm_j &
	\comm{\CLS}{\apm_j} &= \pm\apm_j
\end{align*}
Meanwhile, the four operators \( \cann_j,\ccreate_j \) (\( j=1,2 \)) are ladder operators for \( N \) and \( \CLS \) but commute with \( H \):
\begin{align*}
	\comm{H}{\cpm_j} &= 0 &
	\comm{N}{\cpm_j} &= \pm\cpm_j &
	\comm{\CLS}{\cpm_j} &= \mp\cpm_j 
\end{align*}
The operators \( \cpm_j\) lead to infinite degeneracy.
For each \( j=1,2 \), the three pairs of operators \( \apm_j \), \( \cpm_j \) and \( \spm_j \) mutually commute.

Thus, number operator \( N \) is also the number operator for \( \apm_j,\cpm_j \) and \( \spm_j \).
Indeed,
\[
	N = \sum_{k=1}^2 \left(\frac{1}{2}\acreate_k\aann_k + \frac{1}{2}\ccreate_k\cann_k + \frac{1}{2}\screate_k\sann_k\right).
\]
Thus, we expect that the Fock space in terms of \( \bpm_j,\spm_j \) to be identical to the Fock space in terms of \( \apm_j,\cpm_j,\spm_j \).



Similar to relations~\eqref{eq:1+1interleave}, we have the following relations between the above ladder operators and \( H \) itself:
\begin{align}
	\comm{H}{\apm_j} &= \pm\sqrt{4m\omega}\,\spm_j, &
	\acomm{H}{\spm_j} &= \sqrt{4m\omega}\,\apm_j.
	\label{eq:1+2interleave}
\end{align}

%


\subsubsection{Neither quantum statistics nor parastatistics}\label{sec:1+2notparastatics}
The operators \( \bpm_j \) satisfy the boson relations~\eqref{eq:fermionboson}.
However, despite the following relations
\begin{align*}
	\acomm{\sann_j}{\screate_j} &= \Id, &
	\comm{\aann_j}{\acreate_j} &= \Id, &
	\comm{\cann_j}{\ccreate_j} &= \Id
\end{align*}
the operators \( \spm_j, \apm_j, \cpm_j  \) do not satisfy the fermion nor boson relations~\eqref{eq:fermionboson}
because of the following relations
\[
	\acomm{\sann_1}{\screate_2} = 
	\comm{\aann_1}{\acreate_2} = 
	-\comm{\cann_1}{\ccreate_2} = -2\beta S_0
\]
which should all be zero for fermions\slash bosons.

Moreover, the operators \( \spm_j, \apm_j, \cpm_j \) do not satisfy the parafermion nor paraboson relations of~\eqref{eq:parafermion} and~\eqref{eq:paraboson} because the relations
\begin{align*}
		\comm{\screate_j}{\sann_j} &= -\beta, \\
		\acomm{\acreate_j}{\aann_j} &= N + \CLS + \beta + \frac{1}{2}\Id,\\
		\acomm{\ccreate_j}{\cann_j} &= N - \CLS + \frac{1}{2}\Id
\end{align*}
do not depend on the subscript \( j \)
and hence cannot satisfy the required relations \( \comm*{\comm*{f^+_j}{f^-_j}}{f^-_k} = -2\delta_{jk}f^-_{k} \) (for parafermions)
nor \( \comm*{\acomm*{b^+_j}{b^-_j}}{b^-_k} = -2\delta_{jk}b^-_k \) (for parabosons).
Instead, 
these operators satisfy
\begin{align}\label{eq:notparasac}
		\comm{\comm{\screate_j}{\sann_j}}{\sann_k} &= -2\sann_k, &
		\comm{\acomm{\acreate_j}{\aann_j}}{\aann_k} &= -2\aann_k, &
		\comm{\acomm{\ccreate_j}{\cann_j}}{\cann_k} &= -2\cann_k.
\end{align}
More generally,
the \( \spm_j \) operators satisfy relations~\eqref{eq:notparafermion},
whereas the \( \apm_j \) operators and the \( \cpm_j \) operators satisfy relations~\eqref{eq:notparaboson}.

\subsection{The colour Lie superalgebra}
Observe that we can write
\begin{equation}\label{eq:1+2H2comm}
	H^2 = 2m\omega\acomm{\aann_1}{\acreate_1} - 2m\omega\comm{\sann_1}{\screate_1} + m^2\Id
\end{equation}
using only \( j=1 \) ladder operators. (Note that we could equivalently choose to use only \( j=2 \) ladder operators.)

Observe that
\begin{align*}
	\acomm{\sann_1}{\screate_1} &= \Id, &
	\comm{\aann_1}{\acreate_1} &= \comm{\cann_1}{\ccreate_1} = \Id.
\end{align*}

If we restrict to \( j=1 \)
then the pair
\( \sann_1,\screate_1 \) form ordinary fermion creation\slash annihilation operators,
and
\( \aann_1,\acreate_1 \) and \( \cann_1,\ccreate_1 \) form two pairs of ordinary boson creation\slash annihilation operators.
Moreover, each pair mutually commutes.
This yields the boson-fermion Lie superalgebra with the following basis:
\begin{align*}
	0\text{-sector: }& \Id, \aann_1, \acreate_1, \cann_1, \ccreate_1, &
	1\text{-sector: }& \sann_1, \screate_1.
\end{align*}
We will denote this Lie superalgebra by \( \bfa(2|1) \).
We could also enlarge the above Lie superalgebra to include \( \Id, N,\,\CLS \) and \( \beta \) in the \( 0 \)-sector
(and hence also include \( H^2 = 2m\omega(N + \CLS) + (m\omega + m^2)\Id \)) and include \( H \) in the \( 1 \)-sector.

Recall that bosons and fermions automatically satisfy the paraboson~\eqref{eq:paraboson} and parafermion~\eqref{eq:parafermion} equations.
We can also verify that \( \aann_1,\acreate_1,\cann_1,\ccreate_1 \) and \( \sann_1,\screate_1 \) satisfy the relative parabosons relations~\eqref{eq:relativeparaboson}.
The ladder operators \( \sann_1,\screate_1,\aann_1,\acreate_1,\cann_1,\ccreate_1 \) satisfying \cref{eq:paraboson,eq:parafermion,eq:relativeparaboson} generate the \( \Ztzt \)-graded colour Lie superalgebra%
\footnote{We are using the \( \pso \) notation from~\cite{SVdJ2018}. In the notation of~\cite{Tolstoy2014b}, \( \pso(3|4) \) is written as \( \osp(1,2|4,0) \).}
\( \pso(3|4) \)~\cite{Tolstoy2014b}.
In fact, the full set of ladder operators generate the algebra \( \pso(3|4)\oplus \pso(3|4) \).

Indeed, set
\begin{align*}
	\Hsch&= -\frac{1}{2m}\laplacian  + \frac{1}{2}m\omega^2r^2 = N + \frac{1}{2}\beta + \frac{1}{2}, &
		\spm[\zeta]_\pm &= \frac{1}{2}\screate_1 \mp \zeta\frac{i}{2}\screate_2,\\
		H_\pm &= \frac{1}{2}\Hsch \pm\Hsch\beta S_0, &
		\apm[\zeta]_\pm &= \frac{1}{2}\acreate_1 \mp \zeta\frac{i}{2}\acreate_2,\\
		{\CLS}_\pm &= \frac{1}{2}\CLS +\frac{1}{4}\beta \mp \frac{1}{2}L_0, &
		\cpm[\zeta]_\pm &= \frac{1}{2}\ccreate_1 \pm \zeta\frac{i}{2}\ccreate_2,\\
		{\Sc}_\pm &= -\frac{1}{4}\beta \mp \frac{1}{2}S_0
\end{align*}
where \( \zeta\in\{+,-\} \) is interpreted as \( +1 \) and \( -1 \) in the above equations.
Let \( \pso_\pm(3|4) \) denote the algebra spanned by
\begin{align*}
	00\text{-sector: }& H_\pm, {\CLS}_\pm, {\Sc}_\pm,  & 01\text{-sector: }& \aann_\pm\sann_\pm, \acreate_\pm\screate_\pm, \aann_\pm\screate_\pm, \acreate_\pm\sann_\pm,\\
			  &(\aann_\pm)^2, (\acreate_\pm)^2, (\cann_\pm)^2, (\ccreate_\pm)^2,&&\cann_\pm\sann_\pm, \ccreate_\pm\screate_\pm, \cann_\pm\screate_\pm, \ccreate_\pm\sann_\pm, \\
			  &\aann_\pm\cann_\pm, \acreate_\pm\ccreate_\pm, \aann_\pm\ccreate_\pm, \acreate_\pm\cann_\pm,&&\\
	10\text{-sector: }& \aann_\pm, \acreate_\pm, \cann_\pm, \ccreate_\pm, & 11\text{-sector: }& \sann_\pm, \screate_\pm.
\end{align*}
These colour Lie superalgebras satisfy the following \( \pso(3|4) \) relations:
\begin{align*}
	\comm{H_\pm}{C_\pm} &= 0, &
	\comm{H_\pm}{{\Sc}_\pm} &= \comm{{\CLS}_\pm}{{\Sc}_\pm} = 0,\\
	\comm{H_\pm}{\apm[\zeta]_\pm} &= \zeta\omega\apm[\zeta]_\pm, &
	\comm{{\CLS}_\pm}{\apm[\zeta]_\pm} &= \zeta \apm[\zeta]_\pm, \\
	\comm{H_\pm}{\cpm[\zeta]_\pm} &= \zeta\omega\cpm[\zeta]_\pm, &
	\comm{{\CLS}_\pm}{\cpm[\zeta]_\pm} &= -\zeta\cpm[\zeta]_\pm, \\
	\comm{{\CLS}_\pm}{\spm[\zeta]_\pm} &= 0, &
	\comm{{\Sc}_\pm}{\spm[\zeta]_\pm} &= \zeta\spm[\zeta]_\pm,\\
	\comm{{\Sc}_\pm}{\apm[\zeta]_\pm} &= 
	\comm{{\Sc}_\pm}{\cpm_\pm} = 0, &
	\comm{H_\pm}{\spm[\zeta]_\pm} &= 0,\\
	\acomm{\aann_\pm}{\acreate_\pm} &= \frac{1}{\omega} H_\pm + {\CLS}_\pm, &
	\acomm{\cann_\pm}{\ccreate_\pm} &= \frac{1}{\omega} H_\pm - {\CLS}_\pm, \\
	\comm{\sann_\pm}{\screate_\pm} &= -2{\Sc}_\pm,&
	\acomm{\aann_\pm}{\cann_\pm} &= 2\aann_\pm\cann_\pm,\,\text{etc.}\\
	\acomm{\aann_\pm}{\sann_\pm} &= 2\aann_\pm\sann_\pm,\,\text{etc.} &
	\acomm{\cann_\pm}{\sann_\pm} &= 2\cann_\pm\sann_\pm,\,\text{etc.}\\
	\comm{\apm[\eta]_\pm\spm[\zeta]_\pm}{\apm[\xi]_\pm} &= \frac{1}{2}(\xi - \eta) \spm[\zeta]_\pm, &
	\acomm{\apm[\eta]_\pm\spm[\zeta]_\pm}{\spm[\xi]_\pm} &= \frac{1}{2}\abs{\xi - \zeta}\apm[\eta]_\pm,\\
	\comm{\cpm[\eta]_\pm\spm[\zeta]_\pm}{\cpm[\xi]_\pm} &= \frac{1}{2}(\xi - \eta) \spm[\zeta]_\pm, &
	\acomm{\cpm[\eta]_\pm\spm[\zeta]_\pm}{\spm[\xi]_\pm} &= \frac{1}{2}\abs{\xi - \zeta}\cpm[\eta]_\pm,\\
	\comm{\apm[\eta]_\pm\spm[\zeta]_\pm}{\cpm[\xi]_\pm} &= 0, &
	\comm{\cpm[\eta]_\pm\spm[\zeta]_\pm}{\apm[\xi]_\pm} &= 0
\end{align*}
where \( \zeta,\eta,\xi\in\{+,-\} \) are interpreted as \( +1 \) and \( -1 \) in the above equations.
The remaining relations can be obtained via the colour Jacobi identity.

The ladder operators \( \spm[\zeta]_j,\,  \apm[\zeta]_j,\, \cpm[\zeta]_j \) (\( j=1,2 \))
generate the algebra \( \pso_+(3|4) \oplus \pso_-(3|4) \).
Note that the Hamiltonian
\[
	H = \sqrt{4m\omega}(\acreate_1\sann_1 + \aann_1\screate_1) - 2m\Sc
	 = \sqrt{4m\omega}(\acreate_2\sann_2 + \aann_2\screate_2) - 2m\Sc
\]
is an element of this algebra, but is not homogeneous.
We can include the Hamiltonian as a homogeneous element using similar approaches described for \( (1+1) \)-dimensions in \cref{sec:makeHhomogeneous}.
In addition, \( H^2 = 2m\Hsch + 2m\omega\CLS + 2m\omega\Sc + m^2\Id \) is an element of \( \pso_+(3|4)\oplus\pso_-(3|4)\oplus\CC\Id \).



\begin{rmk}

	Since \( (\beta S_0)^2 = \Id/4 \), the eigenvalues of \( \beta S_0\) must be \( \pm (1/2) \).
	Thus, we can decompose the Hilbert space \( \Hilb \) as eigenspaces of \( \beta S_0 \):
	\[
		\Hilb = \Hilb_+ \oplus \Hilb_-
	\]
	so that \( \beta S_0\ket{\psi} = \pm(1/2)\ket{\psi} \) for \( \ket{\psi}\in\Hilb_{\pm} \).
	Notice that
	\begin{align*}
		\spm[\zeta]_\pm &= \spm[\zeta]_1\left(\frac{1}{2}\Id \pm \beta S_0\right), &
		\apm[\zeta]_\pm &= \apm[\zeta]_1\left(\frac{1}{2}\Id \pm \beta S_0\right), &
		\cpm[\zeta]_\pm &= \cpm[\zeta]_1\left(\frac{1}{2}\Id \pm \beta S_0\right)
	\end{align*}
	and that \( (1/2)\Id \pm \beta S_0 \) is the projection operator onto \( \Hilb_\pm \).
	Therefore, the action of \( \pso_\pm(3|4) \) vanishes on \( \Hilb_\mp \).
	Thus, we only need to consider the action of \( \pso_\pm(3|4) \) on \( \Hilb_\pm \).
	In this way, the structure of the colour Lie superalgebra mirrors that of the Hilbert space.
	The Fock space within \( \Hilb_+ \) or \( \Hilb_- \) is then constructed similarly to \( (1+1) \)-dimensions.
\end{rmk}

\section{Dirac oscillator in (1+3)-dimensions}\label{sec:1+3}
The Dirac Oscillator Hamiltonian in \( (1+3) \)-dimensions is as follows
\begin{equation}\label{eq:1+3Hamiltonian}
	H = \alpha_1(-i\partial_x - \beta im\omega x) + \alpha_2(-i\partial_y - \beta im\omega y) + \alpha_3(-i\partial_z - \beta im\omega z) + \beta m
\end{equation}
where \( \alpha_1,\,\alpha_2,\,\alpha_3 \) and \( \beta \) are gamma matrices (elements of the Clifford Algebra \( \Cl{4}(\CC) \)) satisfying
\[
	\acomm{\gamma_j}{\gamma_k} = 2\delta_{jk}
	\qquad
	\text{for}~(\gamma_1,\gamma_2,\gamma_3,\gamma_4) = (\alpha_1,\alpha_2,\alpha_3,\beta).
\]

In \( (1+3) \)-dimensions we must pay more careful attention to angular momentum.
Let \( \vec{L} = \vec{r}\times\vec{p} = (x,y,z)\times(-i\partial_x,-i\partial_y,-i\partial_z) \) be the orbital angular momentum 
and \( \vec{S} = -\frac{i}{4} \vec{\alpha}\times\vec{\alpha} \) be the spin.
In particular,
\begin{equation}\label{eq:1+3spin}
	S_1 = -\frac{i}{2}\alpha_2\alpha_3, \qquad S_2 = -\frac{i}{2}\alpha_3\alpha_1, \qquad S_3=-\frac{i}{2}\alpha_1\alpha_2.
\end{equation}
In the Dirac representation, these become the standard spin matrices (acting on bispinors):
\[
	S_i = 
	\frac{1}{2}
	\begin{pmatrix}
		\sigma_i & 0\\
		0 & \sigma_i
	\end{pmatrix}
	.
\]

As in the \( (1+2) \)-dimensional case,
\( \vec{L} \) and \( \vec{S} \) are not conserved
but the total angular momentum \( \vec{J} = \vec{L} + \vec{S} \) is, i.e.\
\[
	\comm{H}{J_i} = 0
\]
for \( i=1,2,3 \).

The Hamiltonian squares to
\begin{align*}
	H^2 = -\vec{p}\vdot\vec{p} + m^2\omega^2\vec{r}\vdot\vec{r} - 2m\omega\beta \vec{L}\vdot\vec{S} -3m\omega\beta + m^2\Id.
\end{align*}
Notice the spin-orbit coupling term in \( H^2 \),
\[
	\CLS \coloneqq -2\beta\vec{L}\vdot\vec{S} - \beta,
\]
commutes with \( H \).
This spin-orbit coupling term makes the analysis in \( (1+3) \)-dimensions more complicated.

As in \( (1+1) \)- and \( (1+2) \)-dimensions, define
\begin{align*}
	\sann_j &= -\frac{i}{2}(\beta\alpha_j + \alpha_j), &
	\screate_j &= -\frac{i}{2}(\beta\alpha_j - \alpha_j),\\
	\bann_j &= -\sqrt{\frac{m\omega}{2}}x_j - \sqrt{\frac{1}{2m\omega}}\partial_{x_j}, &
	\bcreate_j &= -\sqrt{\frac{m\omega}{2}}x_j + \sqrt{\frac{1}{2m\omega}}\partial_{x_j}
\end{align*}
for \( j=1,2,3 \)
(recalling that \( (x_1,x_2,x_3) = (x,y,z) \))
and let \( \Bann, \Bcreate, \Sann, \Screate \) be the corresponding vector operators.
Contrasting with lower dimensions, none of these twelve operators form ladder operators for \( H^2 \).
However, the corresponding number operator
\[
	N = 
	\Bcreate\vdot\Bann + \frac{1}{3}\Screate\vdot\Sann
\]
commutes with \( H \).
This implies that the eigenstates for \( H \) lie in the corresponding Fock space.
Furthermore,
\begin{equation}\label{eq:H2NLS}
	H^2 = 2m\omega(N + \CLS + \Id) + m^2\Id.
\end{equation}

The following operators commute with \( H \):
\begin{align}
	N &= \Bcreate\vdot\Bann + \frac{1}{3}\Screate\vdot\Sann\label{eq:1+3N},\\
	\CLS &= -2\beta\vec{L}\vdot\vec{S} - \beta,\\
	\vec{S}^2 &= \vec{S}\vdot\vec{S} = \frac{3}{4}\Id,\\
	\vec{J}^2 &= \vec{J}\vdot\vec{J}.
\end{align}
In addition, the following operators commute with \( H^2 \) (but not with \( H \)):
\begin{align}
	\vec{L}^2 &= \vec{L}\vdot\vec{L},\\
	\beta.
	\label{eq:1+3Sc}
\end{align}
Moreover, the operators~\eqref{eq:1+3N}--\eqref{eq:1+3Sc} all mutually commute.

\begin{rmk}\label{rmk:1+3notparastatistics}
	Just as in \( (1+2) \)-dimensions, the \( \spm_k \) operators are neither fermion nor parafermion operators,
	and instead satisfy relations~\eqref{eq:notparafermion}.
\end{rmk}


\subsection{The colour Lie superalgebra}

Consider the operators 
\begin{align*}
	\bannspin=2\Bann\vdot\vec{S},\qquad
	\bcreatespin=2\Bcreate\vdot\vec{S},\qquad
	\sannspin=\frac{2}{3}\Sann\vdot\vec{S},\qquad
	\screatespin=\frac{2}{3}\Screate\vdot\vec{S}.
\end{align*}
Observe that they satisfy the following commutation relations:
\begin{align*}
	\comm{N}{\bannspin} &= -\bannspin, &
	\comm{N}{\bcreatespin} &= \bcreatespin, &
	\comm{N}{\sannspin} &= -\sannspin, &
	\comm{N}{\screatespin} &= \screatespin, \\
	\comm{\beta}{\bannspin} &= 0, &
	\comm{\beta}{\bcreatespin} &= 0, &
	\acomm{\beta}{\sannspin} &= 0, &
	\acomm{\beta}{\screatespin} &= 0, \\
					    &&&&
	\comm{\vec{L}^2}{\sannspin} &= 0, &
	\comm{\vec{L}^2}{\screatespin} &= 0
\end{align*}
and all commute with \( \vec{J}^2 \) and all anticommute with \( \CLS \).
Moreover,
we have \( H = \sqrt{2m\omega}(\bcreatespin\sannspin + \bannspin\screatespin) - 2m\Sc \).

Using the above relations, these ladder operators send a simultaneous eigenstate of \( N \) and \( \CLS \) to another simultaneous eigenstate of \( N \) and \( \CLS \).
Therefore, \( \bpmspin \) and \( \spmspin \) send \( H^2 \)-eigenspaces to \( H^2 \)-eigenspaces.

The ladder operators \( \spmspin \) satisfy the fermion relations
but the ladder operators \( \bpmspin \) do \emph{not} satisfy the boson relations.
Instead, \( \bpmspin \) are paraboson operators, satisfying~\eqref{eq:paraboson}.
The two series of ladder operators \( \spmspin \) and \( \bpmspin \) commute,
but do not satisfy the relative parafermion~\eqref{eq:relativeparafermion}, nor the relative paraboson~\eqref{eq:relativeparaboson} relations.

Set \( \Sc = -(1/2)\beta \).
To realise both the parafermion relations~\eqref{eq:parafermion} for \( \spmspin \) (recalling that fermion operators satisfy the parafermion relations) and the paraboson relations~\eqref{eq:paraboson} for \( \bpmspin \),
we can use the Lie superalgebra with the following basis:
\begin{align*}
	0\text{-sector: }& N+\Id, \Sc, \sannspin,\screatespin, (\bannspin)^2,(\bcreatespin)^2, &
	1\text{-sector: }& \bannspin,\bcreatespin.
\end{align*}
This Lie superalgebra is isomorphic to \( \osp(1|2)\oplus\spl(2) \).

To instead realise the fermion relations~\eqref{eq:fermionboson} for \( \spmspin \) and the paraboson relations~\eqref{eq:paraboson} for \( \bpmspin \)
requires a \( \Ztzt \)-graded colour Lie superalgebra spanned by the following basis:
\begin{align*}
	00\text{-sector:}&\, N-\Sc, \Id, (\bannspin)^2, (\bcreatespin)^2, &
	01\text{-sector:}&\, \bannspin,\,\bcreatespin, \\
	10\text{-sector:}&\, \sannspin,\, \screatespin, &
	11\text{-sector:}&\, \text{zero}.
\end{align*}
Note that any Lie superalgebra can be given a \( \Ztzt \)-grading
by choosing an embedding of \( \Ztwo \) into \( \Ztzt \)
(doing so will make two of the \( \Ztzt \)-graded sectors empty).
With this in mind,
the above colour Lie superalgebra is isomorphic to
\[
	\osp_{01}(1|2) \oplus \spl_{10}(1|1)
\]
where \( \osp_{01} \) is the orthosymplectic Lie superalgebra with \( \Ztwo = \{00,01\} \)-grading and
\( \spl_{10} \) is the special linear Lie superalgebra with \( \Ztwo = \{00,10\} \)-grading.
Note that, despite the colour Lie superalgebra being a direct sum of Lie superalgebras,
it is not a Lie superalgebra itself due to gradings of the different direct summands having different embeddings into \( \Ztzt \).

The advantage of considering the colour Lie superalgebra is that we can include \( H^2 \) as an element.
To do so, we enlarge the algebra by adding \( \Sc \) to the \( 00 \)-sector and \( \CLS \) to the \( 11 \)-sector,
forming a colour Lie superalgebra isomorphic to
\[
	\osp_{01}(1|2) \oplus \gl_{10}(1|1) \oplus \labelian_{11}
\]
(where \( \labelian_{11} \) is a one dimensional abelian Lie algebra with a \( 11 \)-grading).
Then, \( H^2 = 2m\omega(N + \CLS + \Id) + m^2\Id \) is an element of this enlarged algebra.
This element \( H^2 \) is not homogeneous,
but satisfies the following relations
\begin{align*}
	\cbrak{H^2}{\bannspin} &= -2m\,\bannspin, &
	\cbrak{H^2}{\bcreatespin} &= 2m\,\bcreatespin, \\
	\cbrak{H^2}{\sannspin} &= -2m\,\sannspin, &
	\cbrak{H^2}{\screatespin} &= 2m\,\screatespin
\end{align*}
where \( \cbrak{\cdot}{\cdot} \) is the colour bracket (realised as either a commutator or anticommutator on homogeneous elements).
Compare this with the corresponding relations for the \( (1+1) \)-dimensional equation~\eqref{eq:1+1H2ladder}.
In fact,
setting \( \underline{H}^2 = 2m\omega(N-\CLS + \Id) + m^2\Id \)
we find that
\begin{equation}\label{eq:braidbracket}
	\cbrak{H^2}{\bannspin}
	= H^2\bannspin - \bannspin\underline{H}^2
\end{equation}
and similarly for \( \bcreatespin,\,\sannspin,\,\screatespin \)
(notice that the coefficient of \( \CLS \) in \( \underline{H}^2 \) is \( -1 \)
whereas the coefficient of \( \CLS \) in \( H^2 \) is \( +1 \)---c.f.~\eqref{eq:H2NLS}).

\begin{rmk}
	We cannot capture the behaviour of the ladder operators with \( N \) and \( \CLS \) with a finite-dimensional Lie superalgebra.
	We must use a colour Lie superalgebra.
	Indeed, to realise \( \acomm{\CLS}{\bannspin} = \acomm{\CLS}{\bcreatespin} = \acomm{\CLS}{\sannspin} = \acomm{\CLS}{\screatespin} = 0 \),
	we would need \( \CLS,\,\bannspin,\,\bcreatespin,\,\sannspin,\,\screatespin \) all lying in the odd sector.
	But then we can generate an infinite number of linearly independent operators:
	\[
		x_1 = \bannspin,\,
		x_2 = \comm*{\acomm*{\bcreatespin}{\sannspin}}{x_1},\,
		x_3 = \comm*{\acomm*{\bannspin}{\screatespin}}{x_2},\,
		x_4 = \comm*{\acomm*{\bcreatespin}{\sannspin}}{x_3},\,
		x_5 = \comm*{\acomm*{\bannspin}{\screatespin}}{x_4},\,
		\ldots
	\]
\end{rmk}


The colour Lie superalgebra \(\osp_{01}(1|2) \oplus \gl_{10}(1|1) \oplus \labelian_{11}\) satisfies the following relations:
\begin{align*}
	\comm{N}{\bpmspin} &= \pm\bpmspin, &
	\comm{N}{(\bpmspin)^2} &= \pm2(\bpmspin)^2, \\
	\comm{N}{\spmspin} &= \pm\spmspin, &
	\comm{\Sc}{\spmspin} &= \pm\spmspin, \\
	\comm{(\bpmspin[\mp])^2}{\bpmspin} &= \pm2\bpmspin[\mp], &
	\acomm{\sannspin}{\screatespin} &= \Id,\\
	\comm{(\bannspin)^2}{(\bcreatespin)^2} &= \mathrlap{2\acomm{\bannspin}{\bcreatespin} = 4N -4\Sc + 4\Id}
\end{align*}
with all other relations zero.


\subsection{Fock space and the spectrum}\label{sec:1+3Fockspace}
In this section, we demonstrate the utility of the paraboson operators and the colour Lie superalgebra by solving the Dirac oscillator within a Fock space.
Solutions to the Dirac oscillator have been computed elsewhere (e.g.~\cite{BMyRNYSB1990}),
so we omit many of the details and instead draw the reader's attention to the steps in our solution which make use of the colour Lie superalgebra relations.

\begin{prop}
	There exists a vacuum state \( \vac \) such that \( \bann_k\vac = \sann_k\vac = S_-\vac = 0 \) for \( k = 1,2,3 \).
\end{prop}
\begin{proof}
	The standard argument shows that there exists a state \( \vactilde \) with \( \bann_k\vactilde=\sann_k\vactilde = 0 \).
	If  \( S_-\vactilde = 0 \), then we can set \( \vac = \vactilde \) and we are done. Otherwise, we set \( \vac = S_-\vactilde \).
	Since the commutator of the annihilation operators with \( S_- \) is a linear combination of annihilation operators
	and since \( (S_-)^2 = 0 \), we have that \( \vac \) is a such a vacuum state.
\end{proof}

We now consider the \( (\osp_{01}(1|2)\oplus\spl_{10}(1|1)) \)-module generated by \( \{\vac,S_+\vac\} \), which we interpret as the Fock space.
(We could have instead considered the \( (\osp(1|2)\oplus\spl(2)) \)-module,
which would yield the same vector subspace. However, as we will see, the colour algebra module is more useful.)

We will first search for \( \vec{L} \)- and \( \vec{J} \)-modules within the Fock space.
Let \( L_{\pm}=L_1\pm iL_2 \), \( S_{\pm} = S_1 \pm iS_2 \) and \( J_{\pm} = L_{\pm} + S_{\pm} \) be raising\slash lowering operators for orbital angular momentum, spin and total angular momentum, respectfully.

The elements
\[
	(\bcreate_1 + i\bcreate_2)^\ell\vac \qquad \text{and} \qquad (\bcreate_1 + i\bcreate_2)^\ell S_+\vac
\]
are both \( L_3 \)-highest weight vectors with \( L_3 \)-weight \( \ell \).
Thus, the Fock space contains the \( \vec{L} \)-module \( V_\ell\oplus V_\ell \) for each \( \ell = 0,1,2,\ldots \).
As a \( \vec{J} \)-module, this \( \vec{L} \)-module becomes \( V_\ell\otimes V_{1/2} \).
Applying the Clebsch--Gordon decomposition, we find a \( \vec{J} \)-module
\[
	V_{\ell+1/2}\oplus V_{\ell - 1/2}
\]
as a submodule of the Fock space
(if \( \ell=0 \), then the \( \vec{J} \)-module is \( V_{1/2} \)).
The highest weight vectors of \( V_{\ell+1/2} \) and \( V_{\ell-1/2} \) are
\[
	(\bcreate_1 + i\bcreate_2)^\ell S_+\vac
	\qquad \text{and} \qquad
	(\bcreate_1 + i\bcreate_2)^\ell \vac
	+(\bcreate_1 + i\bcreate_2)^{\ell-1}\bcreate_3 S_+\vac
\]
respectively.
Notice that the eigenvalue of \( N \) on these subspaces is \( n = \ell \)
(recall that \( \comm{N}{\vec{J}^2}=0 \)).

We wish to raise the eigenvalue of \( N \),
without changing the eigenvalues of \( \vec{L}^2,\,\vec{J}^2 \) and \( J_3 \).
However, we cannot achieve this with either \( \bcreatespin \) (which does not preserve the eigenvalues of \( \vec{L}^2 \)) nor with \( \screatespin \) (which could only raise the eigenvalue of \( N \) at most once, since \( (\screatespin)^2 = 0 \)).
Instead, we define new operators
\[
	\tpm=(\bpmspin)^2\spmspin[\mp] + \spmspin
\]
(note that \( \tpm \) is a spectral rescaling of \( (1/2\sqrt{2m\omega})\acomm{H_0}{\bpmspin} \)).
Observe that \( (\tpm)^2 = (\bpm)^2 \):
\begin{align*}
	(\tpm)^2 
	&= \frac{1}{2}\acomm{\tpm}{\tpm}\\
	&= \frac{1}{2}(\acomm{(\bpmspin)^2\spmspin[\mp]}{(\bpmspin)^2\spmspin[\mp]}
	+2\acomm{(\bpmspin)^2\spmspin[\mp]}{\spmspin}
	+\acomm{\spmspin}{\spmspin})\\
	&= \acomm{(\bpmspin)^2\spmspin[\mp]}{\spmspin}\\
	&= (\bpmspin)^2\acomm{\spmspin[\mp]}{\spmspin} + \comm{(\bpmspin)^2}{\spmspin}\spmspin\\
	&= (\bpmspin)^2.
\end{align*}
Note that the above calculation can be performed within the universal enveloping algebra for the \( \Ztzt \)-graded colour Lie superalgebra \( \osp_{01}(1|2) \oplus \spl_{10}(1|1) \) (where \( \tpm \) is \( 10 \)-graded)
but \emph{cannot} be performed in the universal enveloping algebra for the Lie superalgebra \( \osp(1|2)\oplus\spl(2) \) (where \( \tpm \) is \( 0 \)-graded).

\begin{prop}
	The operator \( \tcreate \) is injective.
	Moreover, \( \tcreate \) preserves the eigenvalues of \( \vec{L}^2 \), \( \vec{J^2} \) and \( J_3 \).
\end{prop}
\begin{proof}
	First, we claim that \( N - \Sc + I \) is injective.
	Indeed, take a non-zero state \( \ket{\psi} \) and consider
	\begin{align*}
		\bra{\psi}(N - \Sc + \Id)\ket{\psi}
			&= \bra{\psi}\left(\Bcreate\vdot\Bann + \frac{3}{2}\right)\ket{\psi}\\
			&= \bra{\psi}\Bcreate\vdot\Bann\ket{\psi} + \frac{3}{2}\braket{\psi}\\
			&>0.
	\end{align*}
	Thus \( N-\Sc+\Id \) sends non-zero states to non-zero states (trivial kernel).

	Similarly,
	we claim that \( (\bcreatespin)^2 \) is injective.
	Indeed, take a non-zero state \( \ket{\psi} \) and observe that
	\[
		\bra{\psi}(\bannspin)^2(\bcreatespin)^2\ket{\psi}
		=\bra{\psi}(\bcreatespin)^2(\bannspin)^2\ket{\psi} + 4\bra{\psi}(N-\Sc+\Id)\ket{\psi}
		>0.
	\]
	This last inequality follows because \( \bra{\psi}(\bcreatespin)^2(\bannspin)^2\ket{\psi}  \) is the norm squared of \( (\bannspin)^2\ket{\psi}  \) (which is non-negative);
	and we have already shown above that \( \bra{\psi}(N-\Sc+\Id)\ket{\psi} > 0 \).
	Thus, \( (\bcreatespin)^2 = (\tcreate)^2 \) is injective.
	Therefore, \( \tcreate \) is injective.

	By direct computation,
	we find that \( \comm*{\vec{L}^2}{(\bcreatespin)^2} = 0 \)
	and that \( \comm*{\vec{L}^2}{\spmspin}=0 \).
	That is, \( \tcreate \) preserves the eigenvalues of \( \vec{L}^2 \).
	Since \( \comm{J_i}{\bcreatespin} = \comm{J_i}{\screatespin} = 0 \),
	the eigenvalues of \( \vec{J^2} \) and \( J_3 \) are also preserved by \( \tcreate \).
\end{proof}

By repeatedly applying \( \tcreate \) to a basis for \( V_{\ell+1/2}\oplus V_{\ell - 1/2} \),
we can form a basis for the Fock space.
Explicitly, such a basis is
\[
	\ket{n,\ell,j,j_3}
	= 
	\begin{cases}
		(\tcreate)^{n-\ell}(J_-)^{j-j_3}(\bcreate_1 + i\bcreate_2)^\ell S_+\vac
		& \text{if}~j = \ell+(1/2)\\
		(\tcreate)^{n-\ell}(J_-)^{j-j_3}((\bcreate_1 + i\bcreate_2)^\ell + (\bcreate_1 + i\bcreate_2)^{\ell-1}\bcreate_3S_+)\vac
		& \text{if}~j = \ell-(1/2).
	\end{cases}
\]
These vectors are obviously linearly independent, having different eigenvalues.
That these eigenvectors span the Fock space can be proven with simple dimensional analysis of the eigenspaces of \( N \).
Note that \( n,\ell(\ell+1),j(j+1),j_3 \) are the eigenvalues of \( N,\vec{L}^2, \vec{J}^2, J_3 \) respectively.
Hence,
\begin{align*}
	n&=0,1,2,\ldots &
	\ell&=0,1,\ldots,n, \\
	j &= \ell\pm(1/2), &
	j_3 &= -j, -j+1,\ldots,j-1,j.
\end{align*}
Note also that the basis vectors have not been normalised.

\begin{prop}\label{prop:1+3opket}
We have the following actions on the above basis vectors.
	\begin{enumerate}[label=(\roman*)]
		\item\label{item:1+3screateket} \(\displaystyle
			\screatespin\ket{n,\ell,j,j_3}
			=
			\begin{cases}
				\ket{n+1,\ell,j,j_3} & \text{if \( n-\ell \) even}\\
				0 & \text{if \( n-\ell \) odd}
			\end{cases}
			\)
		\item\label{item:1+3sannket} \(\displaystyle
			\sannspin\ket{n,\ell,j,j_3}
			=
			\begin{cases}
				0 & \text{if \( n-\ell \) even}\\
				\ket{n-1,\ell,j,j_3} & \text{if \( n-\ell \) odd}
			\end{cases}
			\)
		\item\label{item:1+3s0ket} \( \displaystyle
			\Sc\ket{n,\ell,j,j_z} = -\frac{1}{2}(-1)^{n-\ell}\ket{n,\ell,j,j_z}
			\)
		\item\label{item:1+3bcreateket} \(\displaystyle
			\bcreatespin\ket{n,\ell,\ell\pm\frac{1}{2},j_3}
			=
			\ket{n+1,\ell\pm 1,\ell\pm\frac{1}{2},j_3}
			\)
		\item \label{item:1+3bannket} \(\displaystyle
			\bannspin\ket{n,\ell,j,j_3} = 
			\begin{cases}
				\displaystyle
				(n - j + \frac{1}{2}(-1)^{n-\ell})\ket{n-1,\ell+1,j,j_3} & \text{if}~j=\ell+(1/2)\\
				\displaystyle
				(n + j + 1 + \frac{1}{2}(-1)^{n-\ell})\ket{n-1,\ell-1,j,j_3} & \text{if}~j=\ell-(1/2).
			\end{cases}
		\)
	\end{enumerate}
\end{prop}

To illustrate the use of relations from the \( \Ztzt \)-graded colour Lie superalgebra \( \osp_{01}(1|2) \oplus \spl_{10}(1|1) \), we prove part~\ref{item:1+3screateket}.

\begin{proof}[Proof of \ref{item:1+3screateket}]
	Observe that
	\[
		\acomm{\screatespin}{\tcreate}
		= (\bcreatespin)^2 = (\tcreate)^2
	\]
	and hence
	\[
		\comm{\screatespin}{(\tcreate)^2} = \acomm{\screatespin}{\tcreate}\tcreate - \tcreate\acomm{\screatespin}{\tcreate} = 0,
	\]
	using the colour Jacobi identity.
	Arguing inductively, for \( n-\ell\geq0 \) we find that
	\[
		\cbrak{\screatespin}{(\tcreate)^{n-\ell}} = 
		\begin{cases}
			0 & \text{if \( n-\ell \) even}\\
			(\tcreate)^{n-\ell + 1} & \text{if \( n-\ell \) odd}
		\end{cases}
	\]
	where \( \cbrak{\cdot}{\cdot} \) is the bracket of the \( \Ztzt \)-graded colour Lie superalgebra \( \osp_{01}(1|2) \oplus \spl_{10}(1|1) \).

	Observe that \( \comm*{\screatespin}{J_-} = \comm*{\screatespin}{\bcreate_j} = \comm*{\screatespin}{S_+} = 0 \) for \( j=1,2,3 \).
	Furthermore,
	since \( \sannspin\vac = 0 \), we have that
	\[
		\screatespin\vac = ((\bcreatespin)^2\sannspin + \screatespin)\vac = \tcreate\vac.
	\]
	We also have that \( \comm*{\tcreate}{S_+} = \comm*{\tcreate}{\bcreate_j} = \comm*{\tcreate}{J_-} = 0 \).

	First, assume that \( j=\ell+(1/2) \).
	Using the above calculations, assume \( j=\ell+(1/2) \) to find
	\begin{align*}
		\screatespin\ket{n,\ell,j,j_3}
		&= \screatespin(\tcreate)^{n-\ell}(J_-)^{j-j_3}(\bcreate_1 + i\bcreate_2)^\ell S_+\vac\\
		&=
		\begin{cases}
			(\tcreate)^{n-\ell}\screatespin(J_-)^{j-j_3}(\bcreate_1 + i\bcreate_2)^\ell S_+\vac & \text{if \( n-\ell \) even}\\
			-(\tcreate)^{n-\ell}\screatespin(J_-)^{j-j_3}(\bcreate_1 + i\bcreate_2)^\ell S_+\vac \\
			\qquad +(\tcreate)^{n-\ell+1}(J_-)^{j-j_3}(\bcreate_1 + i\bcreate_2)^\ell S_+\vac & \text{if \( n-\ell \) odd}\\
		\end{cases}\\
		&=
		\begin{cases}
			(\tcreate)^{n-\ell}(J_-)^{j-j_3}(\bcreate_1 + i\bcreate_2)^\ell S_+\screatespin\vac & \text{if \( n-\ell \) even}\\
			-(\tcreate)^{n-\ell}(J_-)^{j-j_3}(\bcreate_1 + i\bcreate_2)^\ell S_+\screatespin\vac \\
			\qquad +(\tcreate)^{n-\ell+1}(J_-)^{j-j_3}(\bcreate_1 + i\bcreate_2)^\ell S_+\vac & \text{if \( n-\ell \) odd}\\
		\end{cases}\\
		&=
		\begin{cases}
			(\tcreate)^{n-\ell}(J_-)^{j-j_3}(\bcreate_1 + i\bcreate_2)^\ell S_+\tcreate\vac & \text{if \( n-\ell \) even}\\
			-(\tcreate)^{n-\ell}(J_-)^{j-j_3}(\bcreate_1 + i\bcreate_2)^\ell S_+\tcreate\vac \\
			\qquad +(\tcreate)^{n-\ell+1}(J_-)^{j-j_3}(\bcreate_1 + i\bcreate_2)^\ell S_+\vac & \text{if \( n-\ell \) odd}\\
		\end{cases}\\
		&=
		\begin{cases}
			\ket{n+1,\ell,j,j_3} & \text{if \( n-\ell \) even}\\
			0 & \text{if \( n-\ell \) odd}.
		\end{cases}
	\end{align*}
	If \( j=\ell-(1/2) \) the proof is similar.
\end{proof}

The proofs of parts~\ref{item:1+3sannket}--\ref{item:1+3bannket} follow similarly by direct computation.

Since \( H = \sqrt{2m\omega}(\bcreatespin\sannspin + \bannspin\screatespin) - 2m\Sc \),
we can use~\cref{prop:1+3opket} to compute the energy spectrum and eigenvectors.
\begin{thm}
	\leavevmode
	\begin{enumerate}[label=(\roman*)]
		\item\label{item:1+3oddeigenvector} If \( n-j+(1/2) \) is odd then
			\[
				(E^{\pm}_{n,j} - m)\ket{n,j-(1/2),j,j_3} + \sqrt{2m\omega}\ket{n,j+(1/2),j,j_3}
			\]
			is an energy eigenvector for \( H \) with energy
			\[
				E^{\pm}_{n,j} = \pm\sqrt{2m\omega(n+j) + 3m\omega + m^2}.
			\]
		\item\label{item:1+3eveneigenvector} If \( n-j+(1/2) \) is even 
			then
			\[
				\sqrt{2m\omega}\ket{n,j-(1/2),j,j_3} + (E^{\pm}_{n,j} - m)\ket{n,j+(1/2),j,j_3}
			\]
			is an energy eigenvector for \( H \) with energy
			\[
				E^\pm_{n,j} = \pm\sqrt{2m\omega(n-j)+m\omega+m^2}.
			\]
	\end{enumerate}
\end{thm}

The proof is a direct application of \cref{prop:1+3opket}

\subsection*{Acknowledgments}

This work was supported by the Australian Research Council through Discovery Project DP200101339.
The authors would like to thank the anonymous reviewer for taking the time to read the paper and providing helpful comments.

\appendix
\section{Colour algebra for \texorpdfstring{\( H \)}{H}}\label{sec:1+3Halg}
In this appendix, we present a \( \Ztwo[3] \)-graded colour Lie superalgebra that contains the Hamiltonian.
Contrast this with
the \( \Ztzt \)-graded colour Lie superalgebra used in \cref{sec:1+3}, which does not contain the Hamiltonian (though does contain its square). 

To obtain a colour algebra that closes,
we need to split \( H, \, N,\,\bannspin, \) and \(\bcreatespin \) into coordinate components.
Let
\begin{align*}
	H_j &= \sqrt{2m\omega}(\bcreate_j\sann_j + \bann_j\screate_j), &
	N_j &= \bcreate_j\bann_j + \screate_j\sann_j,\\
	\bannspin[j] &= 2\bann_j S_j, &
	\bcreatespin[j] &= 2\bcreate_j S_j 
\end{align*}
for \( j=1,2,3 \).
Then,
\begin{align*}
	H &= \sum_{j}H_j + m\beta, &
	N &= \sum_j N_j + \beta - \Id,\\
	\bannspin &= \sum_j\bannspin[j], &
	\bcreatespin &= \sum_j\bcreatespin[j].
\end{align*}
Set
\[
	\HxH_j = \frac{1}{4m\omega}(\vec{H}\times \vec{H})_j
\]
where \( \vec{H} = (H_1,H_2,H_3) \).
Note that 
\[
	\HxH_j = \sum_{k,\ell}i(\varepsilon_{jk\ell})^2\bcreate_k\bann_\ell S_j.
\]

Consider the \( \Ztwo[3] \)-graded colour Lie superalgebra spanned by the following operators
\begin{align*}
	000\text{-sector:}&\, I,N_1,\,N_2,\,N_3, &
	111\text{-sector:}&\, \beta,\sannspin,\,\screatespin, \\
	001\text{-sector:}&\, H_1,\, \beta H_1, &
	110\text{-sector:}&\, \bannspin[1],\, \bcreatespin[1],\, \HxH_1, L_1 S_1, \\
	010\text{-sector:}&\, H_2,\, \beta H_2, &
	101\text{-sector:}&\, \bannspin[2],\, \bcreatespin[2],\, \HxH_2, L_2 S_2, \\
	100\text{-sector:}&\, H_3,\, \beta H_3, &
	011\text{-sector:}&\, \bannspin[3],\, \bcreatespin[3],\, \HxH_3, L_3 S_3
\end{align*}
This colour Lie superalgebra satisfies the following relations for \( j,k=1,2,3 \):
\begin{align*}
	\acomm{\sannspin}{\screatespin} &= \Id, &
	\comm{\bannspin[j]}{\bcreatespin[j]} &= \Id,\\
	\comm{N_j}{H_k} &= (\delta_{jk} - 1)\beta H_k, &
	\comm{N_j}{\beta H_k} &= (\delta_{jk} - 1) H_k,\\
	\comm{N_j}{\sannspin} &= -\sannspin, &
	\comm{N_j}{\screatespin} &= \screatespin,\\
	\comm{N_j}{\bannspin[k]} &= -\delta_{jk}\bannspin[k], &
	\comm{N_j}{\bcreatespin[k]} &= \delta_{jk}\bcreatespin[k] ,\\
	\comm{N_j}{\HxH_k} &= (1-\delta_{jk})L_kS_k, &
	\comm{N_j}{L_kS_k} &= (1-\delta_{jk})\HxH_k ,\\
	\acomm{H_j}{\sannspin} &= -\acomm{\beta H_j}{\sannspin} = \sqrt{2m\omega}\,\bannspin[j], &
	\acomm{H_j}{\screatespin} &= \acomm{\beta H_j}{\screatespin} = \sqrt{2m\omega}\,\bcreatespin[j],\\
	\comm{H_j}{\bannspin[j]} &= \comm{\beta H_j}{\bannspin[j]} = -\sqrt{2m\omega}\,\sannspin, &
	\comm{H_j}{\bcreatespin[j]} &= -\comm{\beta H_j}{\bcreatespin[j]} = \sqrt{2m\omega}\,\screatespin,\\
	\acomm{H_j}{H_j} &= \acomm{\beta H_j}{\beta H_j} = 4m\omega N_j, &
	\comm{L_jS_j}{\HxH_j} &= \frac{1}{2}(N_{j+1} - N_{j+2}) 
\end{align*}
as well as the following relations for \( j\neq k \):
\begin{align*}
	\comm{H_j}{H_k} &= -\comm{\beta H_j}{\beta H_k} = 4m\omega \sum_{\ell}\varepsilon_{jk\ell}\HxH_\ell, \\
	\comm{H_j}{\beta H_k} &= 4m\omega \sum_{\ell}\varepsilon_{jk\ell}L_\ell S_{\ell}, \\
	\acomm{H_j}{\HxH_k} &= -\acomm{\beta H_j}{L_kS_k} = \frac{1}{2}\sum_{\ell}\varepsilon_{jk\ell}\beta H_\ell, \\
	\acomm{H_j}{L_kS_k} &= -\acomm{\beta H_j}{\HxH_k} = -\frac{1}{2}\sum_{\ell}\varepsilon_{jk\ell} H_{\ell}, \\
	\acomm{\HxH_j}{\bannspin[k]} &= -\acomm{L_jS_j}{\bannspin[k]} = \frac{1}{2}\sum_{\ell}\varepsilon_{jk\ell}\bannspin[\ell], \\
	\acomm{\HxH_j}{\bcreatespin[k]} &= \acomm{L_jS_j}{\bcreatespin[k]} = -\frac{1}{2}\sum_{\ell}\varepsilon_{jk\ell}\bcreatespin[\ell], \\
	\acomm{L_jS_j}{\HxH_k} &= \frac{1}{2} \sum_{\ell}\varepsilon_{jk\ell}\HxH_\ell, \\
	\acomm{L_jS_j}{L_kS_k} &= \acomm{\HxH_j}{\HxH_k} = -\frac{1}{2}\sum_{\ell}\varepsilon_{jk\ell}L_{\ell}S_{\ell} .
\end{align*}
All remaining (anti)commutation relations are zero.

The Hamiltonian \( H \) is not a homogeneous element, but satisfies the following relations
\begin{align*}
	\cbrak{H}{\bannspin} &= -3\sqrt{2m\omega}\,\sannspin, &
	\cbrak{H}{\bcreatespin} &= 3\sqrt{2m\omega}\,\sannspin, \\
	\cbrak{H}{\sannspin} &= \sqrt{2m\omega}\,\bannspin, &
	\cbrak{H}{\screatespin} &= \sqrt{2m\omega}\,\bcreatespin.
\end{align*}
Compare the above relations with~\eqref{eq:1+1interleave} and \eqref{eq:1+2interleave}.

\end{document}